\documentclass[copyright,creativecommons]{eptcs}
\providecommand{\event}{ITRS 2010} 
\usepackage{ucs}
\usepackage[utf8x]{inputenc}
\usepackage{url,%
     latexsym,%
     abcrules,%
     amsmath,%
     amsfonts,%
     amssymb}

\DeclareUnicodeCharacter{824}{\ensuremath{\not}}
\DeclareUnicodeCharacter{8608}{\ensuremath{\twoheadrightarrow}}
\DeclareUnicodeCharacter{8871}{\ensuremath{\models}}
\DeclareUnicodeCharacter{9655}{\ensuremath{\rhd}}
\DeclareUnicodeCharacter{9675}{\ensuremath{\circ}}

\newcommand{\Top}{{\rm T}}
\newcommand{\ok}{{\rm ok}}
\newcommand{\id}{{\rm id}}
\newcommand{\dom}{{\rm dom}}
\newcommand{\FV}{{\mathrm {FV}}}
\newcommand{\HV}{{\mathrm {HV}}}
\newcommand{\SN}{{\mathrm {SN}}}

\newcommand{\comment}[1]{}

\newtheorem{lemma}{Lemma}[section]
\newtheorem{proposition}[lemma]{Proposition}
\newtheorem{thm}[lemma]{Theorem}
\newtheorem{definition}[lemma]{Definition}
\newtheorem{corollary}[lemma]{Corollary}

\def\titlerunning{Relating Church-Style and Curry-Style Subtyping}
\def\authorrunning{A. Compagnoni \& H. Goguen}

\begin{document}

\title{Relating Church-Style and Curry-Style Subtyping}
\author{Adriana Compagnoni
  \institute{Stevens Institute of Technology \\
        Castle Point on Hudson \\
        Hoboken NJ 07030 USA}
        \email{abc@cs.stevens.edu} \and
      Healfdene Goguen
      \institute{Google Inc. \\
        111 8th Ave. \\
        New York, NY 10011 USA}
      \email{hhg@google.com}
      }

\maketitle



\begin{abstract}
  Type theories with higher-order subtyping or singleton types are
  examples of systems where computation rules for variables are
  affected by type information in the context.  A complication for
  these systems is that bounds declared in the context do not interact
  well with the logical relation proof of completeness or termination.
  This paper proposes a natural modification to the type syntax for
  $F^{ω}_{≤}$, adding a variable's bound to the variable type
  constructor, thereby separating the computational behavior of the
  variable from the context.  The algorithm for subtyping in $F^ω_≤$
  can then be given on types without context or kind information.  As
  a consequence, the metatheory follows the general approach for type
  systems without computational information in the context, including
  a simple logical relation definition without Kripke-style indexing
  by context.  This new presentation of the system is shown to be
  equivalent to the traditional presentation without bounds on the
  variable type constructor.
\end{abstract}

\section{Introduction}

Logical relations are a powerful technique for proving metatheoretic
properties of type theories.  The traditional approach to the
metatheory of type theories, for example that of Pure Type Systems
\cite{barendregt92lambda}, studies properties of untyped reduction and
conversion, and then completes the study of type-checking by proving
strong normalization with a logical relation construction.

This approach has been difficult to adapt to systems where variables
may behave differently according to context information.  Examples
where this may occur are type systems with singleton types, where a
variable of a singleton type is equal to the unique element of the
type, or subtyping, where a type variable may be replaced by its bound
in a derivation of subtyping.  The key difficulty is that strong
normalization of a term or termination of subtyping on higher-order
types depends on information in the context, but that normalization or
termination also needs to be closed under replacement by equal
contexts, in order to model the constructors that introduce the
computational information into the context.

For example, in $F^ω_≤$, consider a putative proof of strong
normalization in the case of a derivation of $X ≤ A : ⋆ → ⋆ ⊢ X(C) :
⋆$.  Such a proof would have a hypothesis that $A(C)$ is strongly
normalizing, since the model must allow a variable to be replaced by
its bound, referred to as promotion.  However, to model the rule that
$⊢ ∀X ≤ A : ⋆ → ⋆. X(C) = ∀X ≤ B : ⋆ → ⋆. X(C) : ⋆$, we would need
that $B(C)$ is strongly normalizing for arbitrary $B$ such that $⊢ A =
B : ⋆ → ⋆$ \emph{before} constructing the model; the behavior of
$X(C)$ varies according to its context.

Several papers have addressed systems of this type, but each of these
approaches differs from the usual approach to metatheory of type
theories.  Compagnoni and Goguen
\cite{CompagnoniAGoguenH:TOSHOSInfComp, CompagnoniAGoguenH:HOSDec} use
an algorithm where a variable's bound is normalized before promoting
the variable.  This allows context replacement to be proved before the
completeness proof, but it seems to be an odd requirement and was only
introduced to get the proof to work.  Furthermore, the algorithm is
less efficient than an algorithm that postpones normalization of the
bound.  Stone and Harper \cite{StoneHarper:ExtensionalEquiv} prove
termination for an algorithm for singletons using the unnormalized
singleton rather than normalizing it first.  Their Kripke-style proof
indexes the model with sets of possible contexts in which a term is
well-typed.  In the example above, the possible $B$ are limited by
considering contexts that arise from bounds introduced by the $∀$
constructor.  This differs from the standard Kripke-style proof of
strong normalization, which is relative to a single context.

In this paper we propose separating the computational behavior of
variables from the context.  We introduce a modified type structure
for $F^ω_≤$ \cite{Cardelli90:Fomegasub, CardelliLongo90, Mitchell90},
where the type constructor for variables is $X_A$ with the variable's
bound $A$ explicitly mentioned.  We call this presentation ``à la
Church'' for its obvious similarity to type labels on
$λ$-abstractions, and we call the traditional presentation ``à la
Curry''.  With this change, the term structure tells us how promotion
will be used without reference to the context.

This presentation allows us to give a kind- and context-free
definition of the algorithm for subtyping, since the only use of the
context in the traditional algorithm is when a variable is replaced
with its bound.  This in turn leads to an approach to the metatheory
consistent with the usual approach for type theories, since promoting
a variable to a type convertible with its bound, the cause of the
difficulties in the system without bounded variables, is never
necessary.  In our example above, the terms would be $∀X ≤ A : ⋆ →
⋆. X_A(C) = ∀X ≤ B : ⋆ → ⋆. X_B(C)$: the behavior of $X_A(C)$ and
$X_B(C)$ is fixed regardless of context.

While changing the term structure could be considered a syntactic
trick, in our opinion our presentation points to a deficiency in the
syntax of the traditional term structure of $F^ω_≤$.  In general,
model constructions work best when there is a close relationship
between terms and derivations: this is best illustrated by Streicher's
extended term structure and partial interpretation for the Calculus of
Constructions \cite{Streicher:book}.  The inability to construct
traditional models to show decidability of higher-order subtyping for
the expected algorithm suggests that the type structure of $F^ω_≤$ is
inappropriate.  We believe that the trick is that the syntax without
bounded type variables works at all.  The equivalence of the two
presentations shows that the additional information necessary for the
model construction can be ignored in programs.

The system Full $F^ω_≤$, with contravariance in the bounds of
quantified types, further illustrates our point.  Surprisingly, in
contrast to the system with unlabeled type variables, the algorithm
for subtyping for Full $F^ω_≤$ with bounded variables cannot be
defined.  Essentially, when the implicit contravariant type
information in the unlabeled types is made explicit, the side
conditions distinguishing variable reflexivity from promotion cannot
be expressed in a valid inductive definition.  We shall discuss the
technical reasons that the definition fails in more detail when we
define the algorithm for Kernel $F^ω_≤$.  However, the inability to
define the algorithm for Full $F^ω_≤$ over the explicit type structure
gives a strong indication that the unlabeled type structure is an
inadequate representation of types for $F^ω_≤$, rather than the
addition of bound information to the type being a trick.

The correctness of the traditional system à la Curry is a consequence
of the equivalence of the two presentations.  There are two
substantial differences in the treatment of the new system.  First,
because variables mention their bound explicitly, the new presentation
exposes the difference between the operations of renaming $[X←Y]$,
which changes variable names but does not change the bound, and
substitution of a variable, $[Y_B/X]$, which replaces the bound of
$X$.  Secondly, the subtyping judgement is needed in the formulation
of the inference rules for the kinding judgement, for example in the
rule \rn{TVar} for kinding a type variable: this is not necessary in
the traditional presentation.

In this paper we address all of the complications listed above.  We
first study the properties of the subtyping relation, including
completeness and correctness, anti-symmetry, transitivity elimination
and decidability.  We complete our development by showing that our
system is equivalent to the traditional one without bounds.  We ignore
completely the term language, since its metatheory is standard once
decidability of subtyping has been proved.  As such, we do not treat
substitution for bounded variables as occur in $∀$, since this
substitution only occurs in the reduction relation for terms.

\section{Syntax}
\label{sec:syntax}

We now present the term constructors, judgements and rules of
inference for kinding and subtyping in $F^ω_≤$.

\subsection{Syntactic Categories}

The kinds of $F^ω_≤$ are the kind $⋆$ of proper types and the kinds $K
→ K'$ of functions on types and type operators.  We assume an infinite
collection of type variable names $X, Y, Z, ...$.  The types include
variables with explicit bounds $X_A$; the top type $\Top_⋆$; function
types $A → B$; and types $∀X ≤ A : K. B$ of polymorphic functions, in
which the bound type variable $X$ ranges over all subtypes of the
upper bound $A$.  Moreover, like $F^ω$, we allow types to be
abstracted on types, of the form $ΛX : K. A$, and we can apply types
to argument types $A(B)$.  Contexts $Γ, Δ$ are either the empty
context $()$ or extended contexts $Γ, X ≤ A : K$.

We identify types that differ only in the names of bound variables.
We write $A(B₁, ..., B_n)$ for $(A(B₁))...(B_n)$.  If $A ≡ X_C(B₁,
..., B_n)$ then $A$ has head variable $X_C$; we write $\HV(-)$ for the
partial function returning the head variable of a type.  We also
extend the top type $\Top_⋆$ to any kind $K$ by defining inductively
$\Top_{K→K'} = ΛX : K. \Top_{K'}$.  We use $X : K$ as an abbreviation
for $X ≤ \Top_K : K$ in contexts; in this case we say $X$ is a
variable without a bound.

Because type variables are decorated with their bounds, we need to be
careful with our definition of substitution: specifically, a renaming
should be restricted to renaming the variables in the bound $A$ of a
variable $X_A$, as opposed to changing the bound as may occur in a
substitution of $Y_B$ for $X$ in $X_A$.  We therefore define parallel
substitutions $γ, δ$ as either the empty substitution $()$; the
extension of a parallel substitution $γ$ with a renaming of a variable
$X$ by another variable $Y$, written $γ[X←Y]$; or the extension of a
parallel substitution $γ$ with a substitution of a variable $X$ by a
type $A$, written $γ[A/X]$.  We say $γ$ is a renaming if $γ = ()$ or
if $γ = γ₀[X←Y]$ with $γ₀$ a renaming.  We write $\id_Γ$ for the
identity renaming of the type variables declared in $Γ$.

We write $B[γ]$ for the capture-avoiding simultaneous replacement of
each of the variables by its corresponding value, defined as follows
on variables and lifted in the usual way to arbitrary types:
\begin{itemize}
\item $X_A[()] = X_A$.
\item $X_A[γ[X←Y]] = Y_{A[γ[X←Y]]}$.
\item $X_A[γ[Y←Z]] = X_{A[γ[Y←Z]]}$, if $X ≠ Y$.
\item $X_A[γ[B/X]] = B$.
\item $X_A[γ[B/Y]] = X_{A[γ[B/Y]]}$, if $X ≠ Y$.
\end{itemize}
\noindent Observe that $B$ cannot be a variable $Y$ or $Z$ in the last
two equations, but must instead be a bounded type variable $Y_C$ or
$Z_D$.

We also write $B[A/X]$ for the parallel substitution that is the
identity renaming on the free variables in $B$ other than $X$, and the
substitution of $X$ by $A$.  We have standard properties of parallel
substitution, for example that $A[γ][δ] ≡ A[γ ○ δ]$ and $(γ[A/X]) ○ δ
= (γ ○ δ)[Aδ/X]$.  We also write $A ▷ B$ for the standard notion of
one-step $β$-reduction.  We have the standard property of
Church--Rosser for reduction.

\subsection{Judgements and Rules of Inference}

The judgement forms are $Γ ⊢ A : K$ for well-kinded types and $Γ ⊢ A ≤
B : K$ for subtyping.  We sometimes write $Γ ⊢ \ok$ for $Γ ⊢ \Top_⋆ :
⋆$, formalizing the well-formedness of $Γ$, and $Γ ⊢ A = B : K$ for $Γ
⊢ A ≤ B : K$ and $Γ ⊢ B ≤ A : K$.  We may also use the metavariable
$J$ to range over statements (right-hand sides of judgements) of any
of these judgement forms.

The rules of inference are presented as simultaneously defined
inductive relations over the judgements.  We start with several
admissible structural rules, and follow with the kinding and subtyping
rules.

\subsubsection{Kinding Rules}

The following rules formalize the judgement $Γ ⊢ A : K$, stating that
the type $A$ is well-formed and of kind $K$ in context $Γ$.

{\small
\infax[TopEmp]
{() ⊢ \Top_⋆ : ⋆}

\infrule[TopExt]
{Γ ⊢ A : K \andalso X ∉ \dom(Γ)}
{Γ, X ≤ A : K ⊢ \Top_⋆ : ⋆}

\infrule[TVar]
{Γ ⊢ B : K \andalso Γ ⊢ A ≤ B : K \andalso X ≤ A : K ∈ Γ}
{Γ ⊢ X_B : K}

\infrule[TAbs]
{Γ, X : K ⊢ A : K'}
{Γ ⊢ ΛX : K. A : K → K'}

\infrule[TApp]
{Γ ⊢ A : K → K' \andalso Γ ⊢ B : K}
{Γ ⊢ A(B) : K'}

\infrule[Arrow]
{Γ ⊢ A : ⋆ \andalso Γ ⊢ B : ⋆}
{Γ ⊢ A → B : ⋆}

\infrule[All]
{Γ, X ≤ A : K ⊢ B : ⋆}
{Γ ⊢ ∀X ≤ A : K. B : ⋆}
}

Notice that in rule \rn{TVar} it is possible that $X ∈ \FV(B)$ when $Γ
⊢ X_B : K$.  The premise that $Γ ⊢ B : K$ in this rule ensures that
Subject Reduction goes through smoothly, without needing to refer to
subtyping.

\subsubsection{Subtyping Rules}

Finally, the following rules formalize the judgement $Γ ⊢ A ≤ B : K$,
stating that type $A$ is a subtype of type $B$ and both are
well-formed of kind $K$ in context $Γ$.

{\small
\infrule[S-Refl]
{Γ ⊢ A : K}
{Γ ⊢ A ≤ A : K}

\infrule[S-Trans]
{Γ ⊢ A ≤ B : K \andalso Γ ⊢ B ≤ C : K}
{Γ ⊢ A ≤ C : K}

\infrule[S-Top]
{Γ ⊢ A : K}
{Γ ⊢ A ≤ \Top_K : K}

\infrule[S-TVar]
{Γ ⊢ A ≤ B : K \andalso Γ ⊢ B = C : K \andalso X ≤ A : K ∈ Γ}
{Γ ⊢ X_B ≤ X_C : K}

\infrule[S-Promote]
{Γ ⊢ A ≤ B : K \andalso Γ ⊢ B : K \andalso X ≤ A : K ∈ Γ}
{Γ ⊢ X_B ≤ B : K}

\infrule[S-TAbs]
{Γ, X : K ⊢ A ≤ B : K'}
{Γ ⊢ ΛX : K. A ≤ ΛX : K. B : K → K'}

\infrule[S-TApp]
{Γ ⊢ A ≤ C : K → K' \andalso Γ ⊢ B = D : K}
{Γ ⊢ A(B) ≤ C(D) : K'}

\infrule[S-Arrow]
{Γ ⊢ B₁ ≤ A₁ : ⋆ \andalso Γ ⊢ A₂ ≤ B₂ : ⋆}
{Γ ⊢ A₁ → A₂ ≤ B₁ → B₂ : ⋆}

\infrule[S-All]
{Γ ⊢ A = C : K \andalso Γ, X ≤ A : K ⊢ B ≤ D : ⋆}
{Γ ⊢ ∀X ≤ A : K. B ≤ ∀X ≤ C : K. D : ⋆}

\infrule[S-BetaL]
{Γ, X : K ⊢ A : K' \andalso Γ ⊢ B : K}
{Γ ⊢ (ΛX : K. A)(B) ≤ A[B/X] : K'}

\infrule[S-BetaR]
{Γ, X : K ⊢ A : K' \andalso Γ ⊢ B : K}
{Γ ⊢ A[B/X] ≤ (ΛX : K. A)(B) : K'}
}

\section{The Algorithm}
\label{sec:alg}

In this section we define the algorithm for kinding and subtyping.

First, we define the relations $→_w$ for weak-head reduction and $↠_n$
for reduction to normal form.

\begin{definition}
  \begin{itemize}
  \item $(ΛX : K. A)(B) →_w A[B/X]$.
  \item $A(B) →_w C(B)$ if $A →_w C$.
  \item $\Top_⋆ ↠_n \Top_⋆$.
  \item $X_A(B₁, ..., B_n) ↠_n X_C(D₁, ..., D_n)$ if $A ↠_n C$ and $B_i
    ↠_n D_i$ for $1 ≤ i ≤ n$.
  \item $ΛX : K. A ↠_n ΛX : K. B$ if $A ↠_n B$.
  \item $A → B ↠_n C → D$ if $A ↠_n C$ and $B ↠_n D$.
  \item $∀X ≤ A : K. B ↠_n ∀X ≤ C : K. D$ if $A ↠_n C$ and $B ↠_n D$.
  \item $A ↠_n C$ if $A →_w B$ and $B ↠_n C$.
  \end{itemize}
  \noindent We also write $A ↓_n B$ iff there is a $C$ such that $A
  ↠_n C$ and $B ↠_n C$, and $A ↓_n$ iff there is a $C$ such that $A
  ↠_n C$.
\end{definition}

\begin{lemma}
  If $A ↠_n C$ and $A ▷ B$ then $B ↠_n C$.
\end{lemma}

The algorithm has a judgement for kinding, $Γ ⊢_A A : K$, and two
judgements for subtyping, $⊢_A A ≤_W B$ for subtyping weak-head normal
forms, and $⊢_A A ≤ B$ for subtyping arbitrary types.  The judgement
$Γ ⊢_A A : K$ corresponds to type inference: the context $Γ$ and type
$A$ are inputs, and the kind $K$ is an output.  The algorithm for
subtyping is analogous to untyped conversion in the $λ$-calculus: it
is purely a computational relation, without reference to kind
information.  Furthermore, the algorithm for kinding does not refer to
subtyping, because subtyping is used for the term language of $F^ω_≤$
and not for types.

The algorithm is defined by the following rules of inference.  It is
syntax-directed, and it will be shown to be terminating on well-formed
types.  Clearly types do not need to be well-formed to be subjects of
the algorithmic subtyping judgement.  Since algorithmic subtyping
incorporates weak-head reduction, it is also clearly not terminating
in general.

{\small
\infax[AT-Top]
{Γ ⊢_A \Top_⋆ : ⋆}

\infrule[AT-TVar]
{X ≤ A : K ∈ Γ}
{Γ ⊢_A X_A : K}

\infrule[AT-TAbs]
{Γ, X : K ⊢_A A : K' \andalso X ∉ \dom(Γ)}
{Γ ⊢_A ΛX : K. A : K → K'}

\infrule[AT-TApp]
{Γ ⊢_A A : K → K' \andalso Γ ⊢_A B : K}
{Γ ⊢_A A(B) : K'}

\infrule[AT-Arrow]
{Γ ⊢_A A : ⋆ \andalso Γ ⊢_A B : ⋆}
{Γ ⊢_A A → B : ⋆}

\infrule[AT-All]
{\begin{array}{c}
  Γ ⊢_A A₁ : K \\
  Γ, X ≤ A₁ : K ⊢_A A₂ : ⋆
 \end{array} \andalso
 \begin{array}{c}
   \\ X ∉ \dom(Γ)
 \end{array}}
{Γ ⊢_A ∀X ≤ A₁ : K. A₂ : ⋆}

\infrule[AWS-Top]
{\mbox{$\HV(A)$ undefined and $A$ is not an abstraction}}
{⊢_A A ≤_W \Top_⋆}

\infrule[AWS-TVar]
{X_A(B₁, ..., B_n) ↓_n X_C(D₁, ..., D_n)}
{⊢_A X_A(B₁, ..., B_n) ≤_W X_C(D₁, ..., D_n)}

\infrule[AWS-Promote]
{⊢_A A(B₁, ..., B_n) ≤ C \andalso C ̸↓_n X_A(B₁, ..., B_n)}
{⊢_A X_A(B₁, ..., B_n) ≤_W C}

\infrule[AWS-TAbs]
{⊢_A A ≤ B}
{⊢_A ΛX : K. A ≤_W ΛX : K. B}

\infrule[AWS-Arrow]
{⊢_A B₁ ≤ A₁ \andalso ⊢_A A₂ ≤ B₂}
{⊢_A A₁ → A₂ ≤_W B₁ → B₂}

\infrule[AWS-All]
{A₁ ↓_n B₁ \andalso ⊢_A A₂ ≤ B₂}
{⊢_A ∀X ≤ A₁ : K. A₂ ≤_W ∀X ≤ B₁ : K. B₂}

\infrule[AS-Inc]
{A ↠_w C \andalso B ↠_w D \andalso ⊢_A C ≤_W D}
{⊢_A A ≤ B}
}

The rule for subtyping of bounded variables, \rn{S-TVar}, states that
$X_A$ is less than $X_B$ if $A$ and $B$ are equal.  The algorithmic
presentation of this system needs a side condition, independent of the
algorithm, to determine whether to apply rule \rn{AWS-TVar} or rule
\rn{AWS-Promote}: in our presentation of Kernel $F^ω_≤$, this is
whether the left- and right-hand sides are convertible.  On the other
hand, a translation of Full $F^ω_≤$ with bounded variables requires
the premise $Γ ⊢ B ≤ A : K$ in \rn{AWS-TVar} to demonstrate the
equivalence of the explicit and unlabeled type presentations.
However, the algorithm cannot be defined with this premise, because
the negation of this premise in \rn{AWS-Promote} requires a negative
occurrence of the judgement being defined, which is not a valid
inductive definition.

Also, observe that while the comparison $A ̸↓_n X_A(B₁, ..., B_n)$ in
\rn{AWS-Promote} might hold if either side diverges, in practice the
algorithm will always be applied to well-formed terms, which will be
shown to be terminating.

\section{Metatheory}
\label{sec:meta}

In this section we develop the basic metatheory for the algorithm.

We begin with the relations $A >_P B$, formalizing a use of promotion,
and $A >_S B$, formalizing $B$ an immediate subterm of $A$, both for
$A$ weak-head normal.

\begin{definition}
  \hfill
  \begin{itemize}
  \item $X_A(B₁, ..., B_n) >_P A(B₁, ..., B_n)$.
  \item
    \begin{itemize}
    \item $ΛX : K. A >_S A$.
    \item $A₁ → A₂ >_S A₁$ and $A₁ → A₂ >_S A₂$.
    \item $∀X ≤ A₁: K. A₂ >_S A₂$.
    \end{itemize}
  \end{itemize}
\end{definition}

\begin{definition}[Strong Normalization, Termination]
  We define the following predicates inductively:
  \begin{itemize}
  \item $\SN(A)$ iff $\SN(B)$ for all $B$ such that $A ▷ B$.
  \item $T(A)$ iff $T(B)$ for all $B$ such that $A ▷ B$, $A >_P B$ or
    $A >_S B$.
  \end{itemize}
\end{definition}

The predicate $T(A)$, or $A$ is terminating, formalizes the possible
types that the algorithm may encounter when invoked on a judgement
containing $A$.  As for strong normalization, a base case for $T(A)$
would be a type with no reducts.

\begin{lemma}
  \label{l:T}
  We have the following properties of $T(-)$ and $\SN(-)$:
  \begin{enumerate}
  \item $T(A)$ implies $\SN(A)$. \label{l:T:SN}
  \item $T(A)$ implies $A ↓_n$. \label{l:T:downarrow}
  \item If $T(A)$ and $A ▷ B$ then $T(B)$. \label{l:T:Red}
  \item $T(\Top_⋆)$. \label{l:T:Top}
  \item $T(A)$ iff $T(ΛX : K. A)$. \label{l:T:Lambda}
  \item $T(\Top_K)$. \label{l:T:TopK}
  \item $T(A)$ and $T(B)$ iff $T(A → B)$. \label{l:T:Arrow}
  \item $\SN(A)$ and $T(B)$ iff $T(∀X ≤ A : K. B)$. \label{l:T:Forall}
  \item If $A >_P B$ then $T(A)$ iff $T(B)$. \label{l:T:promote}
  \item $\SN(A_i)$ for $1 ≤ i ≤ n$ iff $T(X_{\Top_K}(A₁, ...,
    A_n))$. \label{l:T:Var}
  \item If $T(A)$, $T(B)$, $A(B) →_w C$ and $T(C)$ then
    $T(A(B))$. \label{l:T:WH}
  \end{enumerate}
\end{lemma}
\begin{proof}
  The only case that is difficult is Case~\ref{l:T:WH}, which follows
  by standard $λ$-calculus properties.
  \comment{
  follows from the sub-lemma that if $A →_w C$
  and $A ▷ D$ then $C ≡ D$ or there is an $E$ such that $D →_w E$ and
  $C ▷^* E$.  This sub-lemma can be shown using parallel reduction,
  with the inclusion $▷ ⊆ ⇒$ and $⇒ ⊆ ▷^*$.  The property follows by
  induction on $T(A)$ and $T(B)$, where $A(B) ▷ D$ and $D ≢ C$ implies
  $A ▷ A'$ and $D ≡ A'(B)$ or $B ▷ B'$ and $D ≡ A(B')$, and hence
  $T(D)$ by the sub-lemma and the induction hypothesis.
  }
\end{proof}

\begin{proposition}[Decidability]
  \label{p:decide}
  If $T(A)$ and $T(B)$ then $⊢_A A ≤ B$ and $⊢_A A ≤_W B$ terminate.
\end{proposition}
\begin{proof}
  By induction on the sum of the length of the derivations of $T(A)$
  and $T(B)$.

  There are two cases: either $A$ or $B$ has a weak-head reduct or
  they are both weak-head normal.  In the first case the result
  follows by induction hypothesis.  Otherwise, by inspection:
  \begin{itemize}
  \item $A ≡ \Top_⋆$.  If $B ≡ \Top_⋆$ then $⊢_A A ≤ B$ succeeds,
    otherwise it fails.
  \item $A ≡ X_C(D₁, ..., D_n)$.  If $B ≡ X_E(F₁, ..., F_n)$ then $C
    ↓_n E$ and $D_i ↓_n F_i$ terminate, since $T(A)$ implies $A ↓_n$
    and $T(B)$ implies $B ↓_n$, so if these conditions hold then $⊢_A
    A ≤ B$ succeeds.

    Otherwise, $X_C(D₁, ..., D_n) >_P C(D₁, ..., D_n)$, so $⊢_A C(D₁,
    ..., D_n) ≤ B$ terminates by induction hypothesis, so $⊢_A A ≤ B$
    succeeds or fails as $⊢_A C(D₁, ..., D_n) ≤ B$ does.
  \item $A ≡ A₁ → A₂$.  If $B ≡ \Top_⋆$ then $⊢_A A ≤ B$ succeeds.  If
    $B ≡ B₁ → B₂$ then $⊢_A B₁ ≤ A₁$ and $⊢_A A₂ ≤ B₂$ terminate by
    induction hypothesis, since $A₁ → A₂ >_S A_i$ for $i ∈ \{ 1, 2 \}$
    and similarly for $B$, so $⊢_A A₁ → A₂ ≤ B₁ → B₂$ terminates.

    Otherwise, $⊢_A A ≤ B$ fails.
  \item $A ≡ ∀X ≤ A₁: K. A₂$.  If $B ≡ \Top_⋆$ then $⊢_A A ≤ \Top_⋆$
    succeeds.  If $B ≡ ∀X ≤ B₁ : K. B₂$ then $A₁ ↓_n B₁$ terminates
    because $T(A₁)$ and $T(B₁)$ imply $A₁ ↓_n$ and $B₁ ↓_n$.
    Furthermore, $∀X ≤ A₁ : K. A₂ >_S A_2$ and $∀X ≤ B₁ : K. B₂ >_S
    B_2$, so $⊢_A A₂ ≤ B₂$ terminates by induction hypothesis, and
    $⊢_A ∀X ≤ A₁ : K. A₂ ≤ ∀X ≤ B₁ : K. B₂$ succeeds or fails as $⊢_A
    A₂ ≤ B₂$ does.

    Otherwise, $⊢_A A ≤ B$ fails.
  \item $A ≡ ΛX : K. A_1$.  If $B ≡ ΛX : K. B_1$ then $ΛX : K. A₀ >_S
    A₀$ and $ΛX : K. B₀ >_S B₀$, so $⊢_A A_1 ≤ B_1$ terminates by
    induction hypothesis, and $⊢_A ∀X : K. A₀ ≤ ∀X : K. B₀$ succeeds or
    fails as $⊢_A A_1 ≤ B_1$ does.

    Otherwise, $⊢_A A ≤ B$ fails.
  \end{itemize}
\end{proof}

\begin{lemma}[Reflexivity]
  If $T(A)$ then $⊢_A A ≤ A$.
\end{lemma}
\begin{proof}
  We show $⊢_A A ≤_W A$ for weak-head normal $A$ and $⊢_A A ≤ A$ for all
  $A$, by induction on $T(A)$.

  If $A$ is weak-head normal then the proof proceeds by case analysis.
  For example, suppose $A ≡ X_B(C₁, ..., C_n)$.  Then $\SN(B)$ and
  $\SN(C_i$) for $1 ≤ i ≤ n$, so $B ↓_n B$ and $C_i ↓_n C_i$, so $⊢_A
  X_B(C₁, ..., C_n) ≤_W X_B(C₁, ..., C_n)$.

  For arbitrary $A$, $A ↠_w B$ and $⊢_A B ≤_W B$ immediately if $A$ is
  weak-head normal or otherwise by induction hypothesis.
\end{proof}

\begin{lemma}[Subject Conversion]
  If $⊢_A A ≤ B$, $A ↓_n A'$, and $B ↓_n B'$ then $⊢_A A' ≤ B'$.
\end{lemma}
\begin{proof}
  By induction on derivations of $⊢_A A ≤ B$, using Church--Rosser
  for \rn{AWS-Promote}.
\end{proof}

\begin{lemma}[Normalization]
  If $⊢_A A ≤ B$ then there are $A'$ and $B'$ such that $A ↠_n A'$ and
  $B ↠_n B'$.
\end{lemma}

The following lemma simply states that promotion is always valid, even
if the side condition of \rn{AWS-Promote} is not satisfied.  This is
true because if the side condition is not satisfied then \rn{AWS-TVar}
can be applied.

\begin{lemma}[Promotion]
  If $⊢_A B ≤ C$, $A >_P B$ and there is a $D$ such that $A ↠_n D$
  then $⊢_A A ≤ C$.
\end{lemma}
\begin{proof}
  By Normalization there is a $D$ such that $C ↠_n D$.  If $C ↓_n A$
  then $⊢_A A ≤ C$ by \rn{AWS-TVar}, and otherwise $⊢_A A ≤ C$ by
  \rn{AWS-Promote}.
\end{proof}

\begin{lemma}[Transitivity]
  If $⊢_A A ≤ B$ and $⊢_A B ≤ C$ then $⊢_A A ≤ C$.
\end{lemma}
\begin{proof}
  By induction on derivations, using Normalization and Subject
  Conversion in \rn{AWS-TVar} and Promotion in \rn{AWS-Promote}.
\end{proof}

The length of a derivation $T(A)$ includes uses of reduction.  To
prove Anti-Symmetry, we need a measure that is invariant under
reduction but respects $>_P$.

\begin{definition}
  We define an alternative length measure $|T(A)|$ of a derivation of
  $T(A)$ inductively as:
  \[
   \mathit{max}( \{ |T(B)| \mbox{ for $B$ such
    that } A ▷ B \} ∪ \{ |T(C)| + 1 \mbox{ for $C$ such that }A >_P C
  \} )
  \]
\end{definition}

Observe that $|T(A)|$ does not depend on $>_S$.

\begin{lemma}
  \label{l:length}
  We have the following properties of the predicate $T(A)$ and the
  measure $|T(A)|$ of derivations of $T(-)$:
  \begin{enumerate}
  \item If $T(A)$ and $A ↠_n B$ then $|T(A)| =
    |T(B)|$. \label{l:length:reddnf}
  \item If $T(A)$, $T(B)$ and $A ↓_n B$ then $|T(A)| = |T(B)|$.
    \label{l:length:conv}
  \item If $T(A)$ and $A ▷ B$ then $|T(A)| =
    |T(B)|$. \label{l:length:red}
  \item $|T(\Top_⋆)| = 0$.
  \item If $T(ΛX : K. A)$ then $|T(ΛX : K. A)| = 0$.
  \item If $T(A₁ → A₂)$ then $|T(A₁ → A₂)| = 0$.
  \item If $T(∀X ≤ A₁ : K. A₂)$ then $|T(∀X ≤ A₁ : K. A₂)| = 0$.
  \end{enumerate}
\end{lemma}
\comment{
\begin{proof}
  \begin{enumerate}
  \item We first show the result for parallel reduction $⇒$ by
    induction on $T(A)$ and inversion of $⇒$.  Each case follows by a
    simple lemma stating confluence of parallel reduction with the
    respective relations of $>$, the relation of being an applicative
    subterm, and weak-head reduction.  The result follows because $A
    ↠_n B$ implies $A ⇒^⋆ B$.
  \item By definition of $↓_n$ and Case~\ref{l:length:reddnf}.
  \end{enumerate}
\end{proof}
}

\begin{lemma}[Key Lemma]
  If $T(A)$, $T(B)$ and $⊢_A A ≤ B$ then $|T(A)| ≥ |T(B)|$.
\end{lemma}
\begin{proof}
  By induction on $⊢_A A ≤ B$, using Lemma~\ref{l:length}.
\end{proof}

Specifically, notice that the Key Lemma allows us to prove directly
that $⊢_A A(B₁, ..., B_n) ≤ X_A(B₁, ..., B_n)$ is impossible, since
$|T(X_A(B₁, ..., B_n))| > |T(A(B₁, ..., B_n))|$ by definition.  We use
this fact in the proof of Anti-Symmetry.

\begin{lemma}[Anti-Symmetry]
  If $⊢_A A ≤ B$, $⊢_A B ≤ A$, $T(A)$ and $T(B)$, then $A ↓_n B$.
\end{lemma}
\begin{proof}
  By induction on derivations $⊢_A A ≤ B$ and $⊢_A B ≤ A$.

  We consider two cases:
  \begin{itemize}
  \item \rn{AWS-Promote} is used in deriving $⊢_A A ≤ B$.  This is a
    base case with no use of the induction hypothesis.  By the Key
    Lemma $|T(A(B₁, ..., B_n))| ≥ |T(C)|$, and $|T(C)| ≥ |T(X_A(B₁,
    ..., B_n))|$, so $|T(A(B₁, ..., B_n))| ≥ |T(X_A(B₁, ..., B_n))|$.
    However, $|T(X_A(B₁, ..., B_n))| > |T(A(B₁, ..., B_n))|$ by
    definition, which is a contradiction by trichotomy.
  \item \rn{AWS-TAbs} is used in deriving $⊢_A A ≤ B$ and $⊢_A B ≤ A$.
    Then $A ≡ ΛX : K. A₁$ and $B ≡ ΛX : K. B₁$ for some $A₁$ and $B₁$,
    with $⊢_A A₁ ≤ B₁$ and $⊢_A B₁ ≤ A₁$.  By definition $ΛX : K. A₁
    >_S A₁$ implies $T(A₁)$ and similarly $T(B₁)$, so by induction
    hypothesis $A₁ ↓_n B₁$, and so $ΛX : K. A₁ ↓_n ΛX : K. B₁$.
  \end{itemize}
\end{proof}

\section{Completeness of the Algorithm}
\label{sec:complete}

We now perform the logical relation proof to show completeness and
decidability of the algorithm.

\begin{definition}[Semantic Object]
  A type $A$ is a semantic object at kind $K$, written $SO_K(A)$, iff
  $T(A)$ and $⊢_A A ≤ \Top_K$.
\end{definition}

As for typical strong normalization proofs, this notion of semantic
object does not require well-formedness.

\begin{definition}[Interpretation]
  The interpretations of a kind $K$, $⊧ A ∈ K$ and $⊧ A ≤ B ∈ K$, are
  defined by induction on $K$:
  \begin{itemize}
  \item $⊧ A ∈ ⋆$ iff $SO_⋆(A)$.
  \item $⊧ A ≤ B ∈ ⋆$ iff $SO_⋆(A)$, $SO_⋆(B)$ and $⊢_A A ≤ B$.
  \item $⊧ A ∈ K → K'$ iff $SO_{K→K'}(A)$ and $⊧ A(B) ∈ K'$ for all $B$
    such that $⊧ B ∈ K$.
  \item $⊧ A ≤ B ∈ K → K'$ iff $SO_{K→K'}(A)$, $SO_{K→K'}(B)$, $⊢_A A
    ≤ B$ and $⊧ A(C) ≤ B(C) ∈ K'$ for all $C$ such that $⊧ C ∈ K$.
  \end{itemize}

  The interpretation extends to parallel substitutions $⊧ γ ∈ Γ$ as
  follows:
  \begin{itemize}
  \item $⊧ () ∈ ()$.
  \item $⊧ γ[X←Y] ∈ Γ, X ≤ A : K$ iff $⊧ γ ∈ Γ$.
  \item $⊧ γ[A/X] ∈ Γ, X : K$ iff $⊧ γ ∈ Γ$ and $⊧ A ∈ K$.
  \end{itemize}
\end{definition}

Observe that variables in the context of the form $X ≤ A : K$ where $A
≢ \Top_K$ can only take renamings $[X←Y]$.  These variables are never
subject to substitution, since they cannot become the bound variable of
an abstraction.

\begin{lemma}[Saturated Sets]
  \label{l:sat}
  The following properties hold for $⊧ A ∈ K$ and $⊧ A ≤ B ∈ K$:
  \begin{enumerate}
  \item If $⊧ A ≤ B ∈ K$ then $⊧ A ∈ K$ and $⊧ B ∈ K$. \label{l:sat:correct}
  \item If $⊧ A ∈ K$ then $SO_K(A)$. \label{l:sat:SO}
  \item If $⊧ A ≤ B ∈ K$ then $⊢_A A ≤ B$. \label{l:sat:sub}
  \item If $⊧ A ∈ K$ then $⊧ A ≤ A ∈ K$. \label{l:sat:refl}
  \item $⊧\Top_K ∈ K$. \label{l:sat:top}
  \item If $⊧ A ∈ K$ then $⊧ A ≤ \Top_K ∈ K$. \label{l:sat:Atop}
  \item If $⊧ B ∈ K$ and $A >_P B$ then $⊧ A ∈ K$, and similarly for
    the left- and right-hand sides of $⊧ A ≤ B ∈ K$. \label{l:sat:promote}
  \item If $⊧ B ∈ K$, $T(A)$ and $A →_w B$ then $⊧ A ∈ K$, and
    similarly for the left-hand side of $⊧ A ≤ B ∈
    K$. \label{l:sat:wh}
  \item If $⊧ A' ≤ B' ∈ K$, $A ↓_n A'$, $B ↓_n B'$, $T(A)$ and $T(B)$
    then $⊧ A ≤ B ∈ K$. \label{l:sat:conv}
  \item If $⊧ A ≤ B ∈ K$ and $⊧ B ≤ C ∈ K$ then $⊧ A ≤ C ∈
    K$. \label{l:sat:trans}
  \end{enumerate}
\end{lemma}
\begin{proof}
  By induction on $K$, using for example Reflexivity for
  Case~\ref{l:sat:refl}, Cases~\ref{l:sat:SO} and \ref{l:sat:refl} for
  Case~\ref{l:sat:Atop}, Lemma~\ref{l:T} Case~\ref{l:T:promote}
  and Promotion for Case~\ref{l:sat:promote}, and Transitivity for
  Case~\ref{l:sat:trans}.
  \comment{
  We now consider the cases for $K → K'$:
  \begin{enumerate}
  \item By definition.
  \item By definition.
  \item By definition.
  \item TBD.
  \item By Reflexivity for $⊢_A$ and induction hypothesis.
  \item TBD.
  \item By definition $SO_{K→K'}(A)$, so $⊢_A A ≤ \Top_{K→K'}$.

    Suppose $⊧ C ≤ D ∈ K$ and $⊧ D ≤ C ∈ K$.  By induction hypothesis
    $⊧ A(C) ≤ \Top_{K'} ∈ K'$, and by induction hypothesis
    Case~\ref{l:sat:wh} $⊧ A(C) ≤ \Top_{K→K'}(D) ∈ K'$.  Clearly
    $SO_{K→K'}(\Top_{K→K'})$, so $⊧ A ≤ \Top_{K→K'} ∈ K → K'$ by
    definition.
  \item Suppose $⊧ B ∈ K$.  $A(B) >_P A'(B)$ by definition, so $⊧ A(B)
    ∈ K'$ by induction hypothesis Case~\ref{l:sat:promote}.  Then, by
    Lemma~\ref{l:T} Case~\ref{l:T:promote} $T(A)$, and $⊢_A A' ≤
    \Top_K$ and $A >_P A'$ implies $⊢_A A ≤ \Top_K$ by Promotion, so
    $⊧ A ∈ K → K'$ by definition.

    The binary case $⊧ A ≤ B ∈ K$ also uses Promotion to show $⊢_A A ≤
    B$.
  \item Suppose $⊧ B ∈ K$.  Then $T(B)$ by induction hypothesis
    Case~\ref{l:sat:SO}, and $\SN(A)$ and $\SN(B)$ by
    Lemma~\ref{l:T} Case~\ref{l:T:SN}, so $T(A(B))$ by
    Case~\ref{l:T:Var}.  Also $⊢_A A'(B) ≤ \Top_{K'}$ by
    Case~\ref{l:sat:SO}, so $⊢_A A(B) ≤ \Top_{K'}$ by \rn{AS-Inc}.
    Therefore $⊧ A(B) ∈ K'$ by induction hypothesis
    Case~\ref{l:sat:wh}, and $⊧ A ∈ K → K'$ by definition.
  \item By definition, Transitivity and induction hypothesis.
  \end{enumerate}
  }
\end{proof}

\begin{thm}[Completeness]
  Suppose $⊧ γ ∈ Γ$.  Then:
  \begin{itemize}
  \item If $Γ ⊢ A : K$ then $⊢ A[γ] ∈ K$.
  \item If $Γ ⊢ A ≤ B : K$ and $⊧ A[γ] ≤ B[γ] ∈ K$.
  \end{itemize}
\end{thm}
\begin{proof}
  By induction on derivations.  We consider several cases.
  \begin{itemize}
  \item \rn{TopEmp}.  By Lemma~\ref{l:sat} Case~\ref{l:sat:top}.
  \item \rn{TVar}.  If $γ = γ₀[X←Y]$ then by induction hypothesis $⊧
    B[γ] ∈ K$, and $X_{B[γ]} >_P B[γ]$, so $⊧ X_{B[γ]} ∈ K$ by
    Lemma~\ref{l:sat} Case~\ref{l:sat:promote}.

    If $γ = γ₀[A/X]$ then $⊧ A ∈ K$ by definition.
  \comment{
  \item \rn{TApp}.  By induction hypothesis $⊧ A[γ] ∈ K → K'$ and $⊧
    B[γ] ∈ K$, so $⊧ (A(B))[γ] ≡ (A[γ])(B[γ]) ∈ K'$ by definition.
  }
  \item \rn{TAbs}.  By definition $⊧ γ[X←Y] ∈ Γ, X : K$, so by
    induction hypothesis $⊧ A[γ[X←Y]] ∈ K'$, so $T(A[γ[X←Y]])$ implies
    $T((ΛX : K. A)[γ] ≡ ΛY : K. A[γ[X←Y]])$ by Lemma~\ref{l:T}.

    Furthermore, if $⊧ B ∈ K$ then $⊧ γ[B/X] ∈ Γ, X : K$ by
    definition.  We have $⊧ A[γ[B/X]] ∈ K'$ by induction hypothesis
    and $T(A[γ[B/X]])$ by definition.  Then $T((ΛX : K. A)[γ])$ above
    and $T(B)$ by Lemma~\ref{l:sat}, and $(ΛX : K. A)[γ](B) →_w
    A[γ[B/X]]$, so $T((ΛX : K. A)[γ](B))$ by Lemma~\ref{l:T}
    Case~\ref{l:T:WH}.  Therefore, $⊧ (ΛX : K. A)[γ](B) ∈ K'$, and so
    $⊧ (ΛX : K. A)[γ] ∈ K → K'$.  \comment{
  \item \rn{Arrow}.  By induction hypothesis $⊧ A[γ], B[γ] ∈ ⋆$, so $⊧
    (A → B)[γ] ≡ A[γ] → B[γ] ∈ ⋆$ by Lemma \ref{l:T} and
    definition.
  }
  \item \rn{∀}.  By induction hypothesis $⊧ A[γ] ∈ K$, and $⊧ γ[X←Y] ∈
    Γ, X ≤ A : K$ implies $⊧ B[γ[X←Y]] ∈ ⋆$ by induction hypothesis,
    so $T(A[γ])$ and $T(B[γ[X←Y]])$ by Lemma~\ref{l:sat}.  Then $T((∀X
    ≤ A : K. B)[γ] ≡ ∀X ≤ A[γ] : K.  B[γ[X←Y]])$ by Lemma~\ref{l:T},
    so $⊧ (∀X ≤ A : K. B)[γ] ∈ ⋆$ by definition.

  \item \rn{S-Refl}.  By induction hypothesis $⊧ A ∈ K$, so by
    Lemma~\ref{l:sat} Case~\ref{l:sat:refl} $⊧ A ≤ A ∈ K$.
  \item \rn{S-Trans}.  By induction hypothesis and Lemma~\ref{l:sat}
    Case~\ref{l:sat:trans}.
  \item \rn{S-Top}.  By induction hypothesis and Lemma~\ref{l:sat}
    Case~\ref{l:sat:Atop}.
  \comment{
  \item \rn{S-TVar}.  TBD.
  \item \rn{S-Promote}.  {\bf Rewrite: By induction hypothesis $⊧
      Γ(X)[γ] ≤ Γ(X)[γ] ∈ K$, so by Lemma~\ref{l:sat}
      Case~\ref{l:sat:promote} $⊧ (X_{Γ(X)})[γ] ≡ X_{Γ(X)[γ]} ≤
      Γ(X)[γ] ∈ K$.}
  \item \rn{S-TAbs}.  TBD.
  }
  \item \rn{S-TApp}.  By induction hypothesis $⊧ A[γ] ≤ C[γ] ∈ K → K'$
    and $⊧ B[γ] ≤ D[γ] ∈ K$ and $⊧ D[γ] ≤ B[γ] ∈ K$.  Then $B[γ] ↓_n
    D[γ]$ by Lemma~\ref{l:sat} Case~\ref{l:sat:SO} and Anti-Symmetry,
    and by definition $⊧ (A(B))[γ] ≡ (A[γ])(B[γ]) ≤ (C[γ])(B[γ]) ≡
    (C(B))[γ] ∈ K'$.  We also know $⊧ C[γ] ∈ K → K'$ by
    Lemma~\ref{l:sat} and $⊧ D[γ] ∈ K$, so $⊧ (C(D))[γ] ≡ (C[γ])(D[γ]) ∈
    K'$ by definition and $SO_{K'}((C(D))[γ])$ by Lemma~\ref{l:sat}
    Case~\ref{l:sat:SO}.  Therefore $(C[γ])(B[γ]) ↓_n (C[γ])(D[γ])$
    implies $⊧ (A(B))[γ] ≤ (C(D))[γ] ∈ K'$ by Lemma~\ref{l:sat}
    Case~\ref{l:sat:conv}.
  \comment{
  \item \rn{S-Arrow}.  By induction hypothesis $⊧ B₁[γ] ≤ A₁[γ] ∈ ⋆$
    and $⊧ A₂[γ] ≤ B₂[γ] ∈ ⋆$.  By Lemma~\ref{l:sat}
    Case~\ref{l:sat:SO} and Lemma~\ref{l:T} Case~\ref{l:T:Arrow} we
    have $T((A₁ → A₂)[γ])$ and $T((B₁ → B₂)[γ])$, and so $⊧ (A₁ →
    A₂)[γ], (B₁ → B₂)[γ] ∈ ⋆$.  Furthermore, by Lemma~\ref{l:sat}
    Case~\ref{l:sat:SO} $⊢_A B₁[γ] ≤ A₁[γ]$ and $⊢_A A₂[γ] ≤ B₂[γ]$,
    and so $⊢_A (A₁ → A₂)[γ] ≤ (B₁ → B₂)[γ]$ by \rn{AWS-Arrow}, and $⊧
    (A₁ → A₂)[γ] ≤ (B₁ → B₂)[γ] ∈ ⋆$ by definition.
  \item \rn{S-All}.  TBD.
  \item \rn{S-BetaL}.
  \item \rn{S-BetaR}.
  }
  \end{itemize}
\end{proof}

\begin{lemma}
  $⊧ \id_Γ ∈ Γ$.
\end{lemma}
\begin{proof}
  By induction on $Γ$.
\end{proof}

\begin{corollary}[Termination]
  \label{c:T}
  If $Γ ⊢ A  ≤ B : K$ then $T(A)$ and $T(B)$, and $⊢_A A ≤ B$.
\end{corollary}

\begin{corollary}[Anti-Symmetry]
  If $Γ ⊢ A ≤ B : K$ and $Γ ⊢ B ≤ A : K$ then $A ↓_n B$.
\end{corollary}

\begin{corollary}
  If $Γ ⊢ \Top_K ≤ A : K$ then $A ↓_n \Top_K$.
\end{corollary}

\section{Correctness}
\label{sec:correct}

So far we have not needed any properties of the judgements $Γ ⊢ J$.
We now develop some metatheory for those judgements and use the
results to prove the correctness of the algorithm.

\begin{lemma}[Context]
  \label{l:context}
  \begin{enumerate}
  \item If $Γ ⊢ J$ then $\FV(J) ⊆ \dom(Γ)$.
  \item If $X ≤ A : K ∈ Γ$ and $Γ ⊢ \ok$ then $X ∉ \FV(A)$. 
  \item If $Γ ⊢ J$ then $Γ ⊢ \ok$ as a sub-derivation.
  \item If $Γ, Γ' ⊢ \ok$ then $Γ ⊢ \ok$.
  \end{enumerate}
\end{lemma}

\begin{definition}[Renaming]
  \label{d:renaming} 
  $γ$ is a renaming for $Γ$ in $Δ$ if $Δ ⊢ \ok$, $γ$ is a renaming,
  and $γ(X) ≤ A[γ] : K ∈ Γ$ for each $X ≤ A : K ∈ Γ$.
\end{definition}

\begin{lemma}[Renaming]
  \label{l:renaming}
  If $Γ ⊢ J$ and $γ$ is a renaming for $Γ$ in $Δ$ then $Δ ⊢ J[γ]$.
\end{lemma}

\begin{lemma}
  \label{l:thin-ok}
  If $Γ, Γ' ⊢ \ok$, $Γ ⊢ A : K$ and $X ∉ \dom(Γ, Γ')$ then $Γ, X ≤ A :
  K, Γ' ⊢ \ok$.
\end{lemma}

\begin{proposition}[Replacement]
  If $Γ, X ≤ B : K, Γ' ⊢ J$, $Γ ⊢ A ≤ B : K$ and $Γ ⊢ A : K$ then $Γ,
  X ≤ A : K, Γ' ⊢ J$.
\end{proposition}

\begin{proposition}[Thinning]
  \label{p:thinning} 
  If $Γ, Γ' ⊢ J$, $Γ ⊢ A:K$ and $X ∉ \dom(Γ, Γ')$ then $Γ, X ≤ A:K, Γ'
  ⊢ J$.
\end{proposition}
\begin{proof}
  By Lemmas~\ref{l:context} and \ref{l:thin-ok} $Γ, X ≤ A : K, Γ' ⊢
  \ok$.  Observe that $\id_Γ$ is a renaming for $Γ, Γ'$ in $Γ, X ≤ A :
  K, Γ'$. Then $Γ, X ≤ A:K, Γ' ⊢ J$ by Renaming.
\end{proof}

\begin{proposition}[Substitution]
  \label{l:subst}
  If $Γ, X : K, Γ' ⊢ J$ and $Γ ⊢ A : K$ then $Γ, Γ'[A/X] ⊢ J[A/X]$.
\end{proposition}

\comment{

{\bf Prove Thinning and Substitution simultaneously.}

\begin{definition}[Parallel Reduction]
  Parallel reduction, $A ⇒ B$, is defined inductively.
\end{definition}

\begin{lemma}
  \label{l:parred}
  Parallel reduction has the following properties:
  \begin{itemize}
  \item $A ⇒ A$. \label{l:parred:refl}
  \item If $A ⇒ A'$ and $B ⇒ B'$ then $A[B/X] ⇒
    A'[B'/X]$. \label{l:parred:subst}
  \item If $A ▷ B$ then $A ⇒ B$. \label{l:parred:redinc}
  \item If $A ⇒ B$ then $A ▷^* B$. \label{l:parred:incred}
  \end{itemize}
\end{lemma}
}

\begin{lemma}[Subject Reduction]
  If $Γ ⊢ A : K$ and $A ▷ B$ then $Γ ⊢ A = B : K$.
\end{lemma}
\begin{proof}
  By induction on derivations.
\end{proof}
\comment{
\begin{proof}
  {\bf To add some cases here for submitted version.}
\end{proof}
}

{\bf {Add generation.}}

\begin{proposition}[Correctness]
  The algorithm is correct for the declarative judgements:
  \begin{itemize}
  \item If $Γ ⊢ ok$ and $Γ ⊢_A A : K$ then $Γ ⊢ A : K$.
  \item If $Γ ⊢ A, B : K$ and $⊢_A A ≤_W B$ then $Γ ⊢ A ≤ B : K$.
  \item If $Γ ⊢ A, B : K$ and $⊢_A A ≤ B$ then $Γ ⊢ A ≤ B : K$.
  \end{itemize}
\end{proposition}
\begin{proof}
  By induction on derivations, using Context and Renaming in
  \rn{AT-TVar}; the generation property and Subject Reduction for
  \rn{AWS-TVar}; the generation property for \rn{AWS-Top} and
  \rn{AWS-Promote}; Subject Reduction and Context Replacement in
  \rn{AWS-All}, and Subject Reduction for \rn{AS-Inc}.
\end{proof}

\begin{corollary}[Decidability of Subtyping]
  Subtyping is decidable.
\end{corollary}
\begin{proof}
  Suppose $Γ ⊢ A, B : K$.  By Corollary~\ref{c:T} $T(A)$ and $T(B)$,
  and so $⊢_A A ≤ B$ is decidable by Proposition~\ref{p:decide}, and
  so by Correctness, $Γ ⊢ A ≤ B : K$ is also decidable.
\end{proof}

\section{Relationship with Traditional $F^ω_≤$}
\label{sec:equiv}

We now show that the system with the bounded variable constructor is
equivalent to the traditional presentation of the system, without
bounds in the variable constructor.

Explicitly, we take the syntax of traditional $F^ω_≤$
\cite{Compagnoni:thesis, SteffenPierce97}, which differs from the
syntax presented here only by having a type constructor $X$ instead of
the bounded type constructor $X_A$.  We write judgements in this
system as $Γ ⊢_T J$, with the decoration $T$ for traditional.  The
rules of inference are also standard: we include two rules here but
refer the reader to the standard references for the complete
system.\footnote{However, \cite{Compagnoni:thesis} has intersection
  types rather than an explicit rule for $\Top_⋆$.}

{\small
\infrule[TS-Conv]
{Γ ⊢_T A, B : K \andalso A =_β B}
{Γ ⊢_T A ≤ B}

\infrule[TS-TApp]
{Γ ⊢_T A ≤ B \andalso Γ ⊢_T A(C) : K}
{Γ ⊢_T A(C) ≤ B(C)}
}

\begin{definition}[Decoration]
  Decoration of a type, $A^Γ$, where $Γ$ is a context in the system
  with bounded variables, is a partial function that maps types from
  the Curry-style presentation to the more explicit structure à la
  Church:
  \begin{itemize}
  \item $X^Γ = X_B$, if $X ≤ B : K ∈ Γ$.
  \item $(A → B)^Γ = A^Γ → B^Γ$.
  \item $(∀X ≤ A : K. B)^Γ = ∀X ≤ A^Γ : K. B^{Γ, X ≤ A : K}$.
  \item $(ΛX : K. A)^Γ = ΛX : K. A^{Γ, X : K}$.
  \item $(A(B))^Γ = A^Γ(B^Γ)$.
  \item $(\Top_⋆)^Γ = \Top_*$.
  \end{itemize}
  The extension to contexts, $Γ^d$, is defined in the obvious way:
  \begin{itemize}
  \item $()^d = ()$.
  \item $(Γ, X ≤ A : K)^d = Γ^d, X ≤ A^{Γ^d} : K$.
  \end{itemize}
\end{definition}

\begin{definition}[Erasure]
  The erasure map is simply the homomorphic extension of the stripping
  of the bound from the variable constructor:
  \begin{eqnarray*}
  |X_A| & = & X
  \end{eqnarray*}
  The extension to contexts is also the homomorphic extension.
\end{definition}

Decoration and erasure have some simple properties, most important of
which is that both preserve $β$-equality:
\begin{lemma}
  \label{l:equiv:beta}
  \begin{itemize}
  \item If $A =_β B$ then $|A| =_β |B|$.
  \item If $A =_β B$ and $A^Γ$, $B^Γ$ defined then $A^Γ =_β B^Γ$.
  \end{itemize}
\end{lemma}

Now, we can relate the Curry and Church presentations of $F^ω_≤$.  The
proofs rely on standard properties of the traditional presentation,
for example Church--Rosser for untyped reduction, Generation
properties, well-formedness of contexts, well-kindedness of subtyping,
and uniqueness of kinds.

\begin{lemma}[Soundness]
  \begin{itemize}
  \item If $Γ ⊢_T A : K$ then $Γ^d$ and $A^{Γ^d}$ are defined and $Γ^d
    ⊢ A^{Γ^d} : K$.
  \item If $Γ ⊢_T A ≤ B$ and $Γ ⊢_T A, B : K$ then $Γ^d$, $A^{Γ^d}$,
    $B^{Γ^d}$ are well-formed and $Γ^d ⊢ A^{Γ^d} ≤ B^{Γ^d}: K$.
  \end{itemize}
\end{lemma}
\begin{proof}
  We consider the rule \rn{TS-Conv}.  By Church--Rosser, there exists
  a $C$ such that $A ▷^* C$ and $B ▷^* C$.  By
  Lemma~\ref{l:equiv:beta}, $A^{Γ^d} ▷^* C^{Γ^d}$ and $B^{Γ^d} ▷^*
  C^{Γ^d}$.  By the induction hypothesis, $Γ^d ⊢ A^{Γ^d} : K$ and $Γ^d
  ⊢ B^{Γ^d}: K$.  By Subject Reduction, $Γ^d ⊢ A^{Γ^d} = C^{Γ^d} : K$
  and $Γ^d ⊢ B^{Γ^d} = C^{Γ^d}: K$, and we have $Γ^d ⊢ A^{Γ^d} ≤
  C^{Γ^d}: K$ and $Γ^d ⊢ C^{Γ^d} ≤ B^{Γ^d} : K$ by definition.
  Finally, by \rn{S-Trans}, $Γ^d ⊢ A^{Γ^d} ≤ B^{Γ^d}:K$.
\end{proof}

\begin{lemma}[Completeness]
  \begin{itemize}
  \item If $Γ ⊢ A : K$ then $|Γ| ⊢_T |A| : K$.
  \item If $Γ ⊢ A ≤ B : K$ then $|Γ| ⊢_T |A| ≤ |B|$ and $|Γ| ⊢_T |A|,
    |B| : K$.
  \end{itemize}
\end{lemma}
\begin{proof}
  We consider the rule \rn{S-TApp}.  By the induction hypothesis, $|Γ|
  ⊢_T |A| ≤ |C|$, $|Γ| ⊢_T |A|, |C| : K → K'$, $|Γ| ⊢_T |B| ≤ |D|$,
  and $|Γ| ⊢_T |B|, |D| : K$.  By the kinding rule for type
  application in the traditional presentation and the definition of
  erasure, we have $|Γ| ⊢_T |AB|, |CB|, |CD| : K'$, and by
  \rn{TS-TApp} $|Γ| ⊢_T |AB| ≤ |CB|$.  By Anti-Symmetry $B =_β D$, and
  $|B| =_β |D|$ by Lemma~\ref{l:equiv:beta}, so $|CB| =_β |CD|$.  By
  \rn{TS-Conv} $|Γ| ⊢_T |CB| ≤ |CD|$, and finally, by Transitivity,
  $|Γ| ⊢_T |AB| ≤ |CD|$.
\end{proof}

The important metatheoretic results for subtyping now transfer
straightforwardly to the traditional one.

\begin{corollary}
  \begin{itemize}
  \item If $Γ ⊢_T A : K$ then $A$ is strongly normalizing.
  \item Typechecking $Γ ⊢_T A : K$ and subtyping $Γ ⊢_T A ≤ B : K$ are
    decidable.
  \item If $Γ ⊢_T A ≤ B$, $Γ ⊢_T B ≤ A$, $Γ ⊢_T A : K$ and $Γ ⊢_T B :
    K$ then $A =_β B$.
  \end{itemize}
\end{corollary}

\section{Related and Future Work}
\label{sec:work}

An earlier version of this article was published in the unrefereed
proceedings of Henk Barendregt's Festschrift \cite{henk60}.  The
current paper extends the results of the earlier version by showing
the equivalence with the traditional presentation of the system.

In an earlier paper \cite{CompagnoniAGoguenH:TOSHOSInfComp}, we
considered an algorithm that reduces types to normal form before
invoking the promotion rule in the algorithm.  This makes context
replacement trivial for equal types, since they have the same normal
form and so altering the context does not alter the path of types
considered by the algorithm.  However, this algorithm is not optimal,
since it normalizes the head earlier than necessary.

That earlier paper also used a typed operational semantics to show
termination of the algorithm.  This gave a more extensive treatment of
the metatheory, and the admissibility of thinning, substitution and
context replacement were consequences of the model.  Furthermore,
Subject Reduction was straightforward in the typed operational
semantics.  In the current paper, finding the exact formulation
necessary to show these results in the declarative system $Γ ⊢ J$
turned out to be somewhat subtle, since the kinding judgement uses the
subtyping judgement for the bounded variable rule.  However, the
approach using typed operational semantics was also longer and less
approachable, and involved Kripke models for the proof of
completeness.  We hope that the current paper is clearer by not
defining an intermediate system.

In separate work \cite{CompagnoniAGoguenH:AntiSymJournal}, we also
proved anti-symmetry of higher-order subtyping using the typed
operational semantics.  The basic idea of that paper was to include
the sub-derivation of replacing the variable in a bounded head
variable expression $X_A(B₁, ..., B_n)$ with its bound, $A(B₁, ...,
B_n)$.  This idea is captured in the current paper by the $T(-)$
predicate.  The $T(-)$ predicate is also similar to Compagnoni's
approach with $+$-reduction \cite{Compagnoni:thesis}, but we do not
need to develop the metatheory of a new reduction relation.

As mentioned in the introduction, Stone and Harper
\cite{StoneHarper:ExtensionalEquiv} use a logical relation defined
over sets of contexts, instead of the standard logical relations over
single contexts, to show termination of an algorithm for a type theory
with singleton types, Σ and Π types, and all of the η rules.  Their
work does not normalize the singleton types.  This is an elegant
solution to the problem of varying contexts, but it raises the
question of why singletons or $F^ω_≤$ should have different
requirements on the Kripke-style relation than other type systems.

Abel \cite{abel:mscs06} has shown equivalence of a subtyping algorithm
for higher kinds with polarity by direct induction on kinds rather
than using a logical relation.

There are several directions for future work.  We would like to show
that a Harper–Pfenning-style algorithm \cite{HarperPfenning:alg} is
correct and complete for the type system.  Furthermore, it would be
nice to be able to prove context conversion and Church--Rosser in the
model, as can be done for logical relations for equality, rather than
proving them for the algorithm and lifting to the model.  However,
properties that follow straightforwardly for equality, such as that $⊧
A = B ∈ K$ implies $⊧ A = A ∈ K$, cannot be shown so easily for
subtyping.  Finally, another candidate type construct that we might
study with our technique of explicit type information is singleton
types, which also have computational behavior expressed in the
context.

\section{Conclusions}
\label{sec:concl}

We have introduced a natural and powerful extension of the syntax of
$F^ω_≤$ and showed that the development of the metatheory is similar
to the standard metatheory for type theories, specifically without a
Kripke-style model and with a simple inductive definition capturing
termination of the algorithm.  We have shown all of the important
results for the system, including anti-symmetry, transitivity
elimination and decidability of subtyping.

{\small
\bibliographystyle{eptcs}
\bibliography{bib}
}

\end{document}